\documentclass[11pt]{article}

\RequirePackage[margin=1.1in]{geometry}
\RequirePackage[T1]{fontenc}
\RequirePackage[utf8]{inputenc}
\RequirePackage{lmodern}
\RequirePackage{microtype}

\RequirePackage{amsmath,amssymb,mathtools}
\RequirePackage{amsthm}
\RequirePackage{mathtools}

\RequirePackage{graphicx}
\RequirePackage{float}
\RequirePackage[table,dvipsnames]{xcolor}
\RequirePackage{array}

\definecolor{GreenPale}{HTML}{E4F1E5}
\newcommand{\plainval}[1]{\rule[-3pt]{0pt}{12pt}\ensuremath{#1}}
\newcommand{\oldval}[1]{\begingroup\setlength{\fboxsep}{1.2pt}\setlength{\fboxrule}{.45pt}\fbox{\rule[-2pt]{0pt}{9.5pt}\ensuremath{#1}}\endgroup}
\newcommand{\newval}[1]{\cellcolor{GreenPale}\plainval{#1}}
\newcommand{\keepval}[1]{\cellcolor{GreenPale}\oldval{#1}}
\newcommand{\uvec}{\mathbf u}

\newif\ifshowcomments
\showcommentstrue

\usepackage[normalem]{ulem} %

\RequirePackage{tikz}
\usetikzlibrary{arrows.meta,backgrounds}

\definecolor{SquareBlue}{HTML}{9FBFDD}
\definecolor{SquareOrange}{HTML}{E9B58E}

\newcommand{\DihedralSquare}[4]{%
  \foreach \xx/\yy/\agent/\border in
    {-1/1/1/SquareBlue,1/1/2/SquareOrange,
      1/-1/3/SquareBlue,-1/-1/4/SquareOrange}{%
    \pgfmathsetmacro{\px}{#1*#3*(\xx*cos(-90*#2)-\yy*sin(-90*#2))}%
    \pgfmathsetmacro{\py}{#1*(\xx*sin(-90*#2)+\yy*cos(-90*#2))}%
    \coordinate (v\agent) at (\px,\py);
  }
  \draw[square edge] (v1)--(v2)--(v3)--(v4)--cycle;
  \draw[square edge] (v1)--(v3) (v2)--(v4);
  \ifnum#4=1
    \foreach \xx/\yy/\good in
      {-1/0/a,0/1/b,1/0/c,0/-1/d,
        -.405/.405/{x,p},.405/.405/{y,q}}{%
      \pgfmathsetmacro{\px}{#1*#3*(\xx*cos(-90*#2)-\yy*sin(-90*#2))}%
      \pgfmathsetmacro{\py}{#1*(\xx*sin(-90*#2)+\yy*cos(-90*#2))}%
      \node[good] at (\px,\py) {$\good$};
    }
  \fi
  \foreach \agent/\border in
    {1/SquareBlue,2/SquareOrange,3/SquareBlue,4/SquareOrange}{%
    \node[agent,draw=\border] at (v\agent) {\agent};
  }
}

\RequirePackage{enumitem}

\RequirePackage[square,numbers]{natbib}
\RequirePackage[
  colorlinks=true,
  linkcolor=blue!45!black,
  citecolor=blue!45!black,
  urlcolor=blue!45!black
]{hyperref}
\RequirePackage[nameinlink,capitalize]{cleveref}
\setcitestyle{numbers,square}

\allowdisplaybreaks
\graphicspath{{figures/}}

\theoremstyle{plain}
\newtheorem{theorem}{Theorem}
\newtheorem*{approxtheoremrestated}{\Cref{thm:approx}}

\newtheorem{openproblem}{Open Problem}

\theoremstyle{definition}
\newtheorem{definition}{Definition}

\theoremstyle{plain}

\newcommand{\eps}{\varepsilon}

\makeatletter
\renewcommand{\maketitle}{%
  \begingroup
  \renewcommand{\thefootnote}{\fnsymbol{footnote}}%
  \begin{center}
    {\LARGE\bfseries \@title\par}
    \vspace{1.25em}
    \begin{NoHyper}{\large \@author\par}\end{NoHyper}
  \end{center}
  \@thanks
  \endgroup
  \setcounter{footnote}{0}%
  \vspace{-0.5em}
}
\makeatother

\definecolor{citationPurple}{rgb}{0.4, 0.3, 0.7}
\definecolor{referenceGreen}{rgb}{0.15, 0.55, 0.1}
\hypersetup{
  citecolor=citationPurple,
  linkcolor=referenceGreen
}

\title{Pareto Inefficiency of EFX}
\author{%
\begin{tabular}{@{}c@{\hspace{2em}}c@{\hspace{2em}}c@{}}
Sepideh Mahabadi
& Prasanna Ramakrishnan
& Jakub Tarnawski \\[3pt]
{\normalsize Microsoft Research}
& {\normalsize Harvard University}
& {\normalsize Microsoft Research} \\[2pt]
{\small\texttt{smahabadi@microsoft.com}}
& {\small\texttt{pras\_r@seas.harvard.edu}}
& {\small\texttt{jatarnaw@microsoft.com}}
\end{tabular}%
}

\begin{document}

\maketitle

\begin{abstract}
We consider the problem of allocating indivisible goods to agents with additive valuations, under two classic notions of fairness and efficiency: envy-freeness up to any good (EFX) and Pareto optimality. While some of the earliest work on EFX showed that Pareto optimal EFX allocations may not exist if agents can have zero value for some goods, it has remained open whether the same is true when all values are positive. 

We present a simple instance in which every good has positive value to every agent, but no Pareto optimal EFX allocation exists. We then generalize to a family of instances showing that Pareto optimal $\alpha$-EFX allocations may not exist for any $\alpha > 2/3$.

\end{abstract}

\section{Introduction}

Envy-freeness is a compelling notion of fairness when a group of agents must divide a collection of goods. As motivating examples, consider heirs dividing an estate, a divorcing couple splitting their possessions, or teams drafting new players in a sports league. In any of these settings, if we can manage to make every participant feel that they got as good of an outcome as anyone else, it would be hard to take any complaints of unfairness seriously. 

Miraculously, envy-freeness is achievable in remarkably general settings, provided the goods are divisible \citep{Varian74,Stromquist80,Woodall80}. But if goods are indivisible, envy-freeness is obviously out of reach: if one good is too valuable, whoever gets it will be envied. To sidestep this issue, researchers have introduced several relaxations that allow envy between agents so long as it is small enough to be alleviated by removing one good. The most permissive of these is \emph{envy-freeness up to one good} (EF1), which allows agent $i$ to envy agent $j$ if the envy is alleviated by removing \emph{some} good from $j$'s bundle, intuitively the most valuable one \citep{LiptonMMS04,Budish11}. EF1 allocations exist even for general monotone valuations and can be found with simple, intuitive algorithms \citep{LiptonMMS04}. With additive valuations, they are also compatible with other desirable properties of fairness and economic efficiency \citep{CaragiannisKMPSW16,CaragiannisKMPSW19}.

A much more demanding relaxation is \emph{envy-freeness up to any good} (EFX), which requires any envy that $i$ has for $j$ to be alleviated after removing \emph{any} good from $j$'s bundle, intuitively the \emph{least} valuable one \citep{GourvesMT14,CaragiannisKMPSW16,CaragiannisKMPSW19}. In contrast to EF1, the existence of EFX allocations when agents have additive valuations over the goods remains one of the central open problems in fair division \citep{Procaccia20,AmanatidisABFLMVW23}.

Since EFX is arguably the closest analogue of envy-freeness for indivisible goods, establishing the existence of EFX allocations is an alluring prospect. But one criticism is that EFX can sometimes force outcomes that seem unintuitive. Consider an example of Plaut and Roughgarden~\citep{PlautR18,PlautR20}, in which two agents $1$ and $2$ must divide goods $a$, $b$, and $c$, and their values for the goods are given by \Cref{tab:efx-po-zero}.

\begin{table}[htbp]
\centering
\smallskip
\setlength{\tabcolsep}{11pt}
\renewcommand{\arraystretch}{1.13}
\begin{tabular}{c|rrr}
 & $a$ & $b$ & $c$\\\hline
1 & 2 & 1 & 0\\
2 & 2 & 0 & 1
\end{tabular}
\caption{A simple instance with no Pareto optimal EFX allocation, if zero values are allowed~\citep{PlautR18,PlautR20}.}\label{tab:efx-po-zero}
\end{table}

A reasonable instinct is to break the tie for $a$ arbitrarily, and then give $b$ and $c$ to the agents who value them. However, these outcomes fail EFX. For example, if agent 1 gets $a$ and $b$ and agent 2 gets $c$, then agent 2 envies agent 1 even with $b$ removed. To achieve EFX, we must give one agent $a$ and the other agent \emph{both} $b$ and $c$. But then the latter agent has a good that she does not value at all, and that the other agent would gladly take.  It seems reasonable to hope that a fair allocation would avoid this kind of waste. This is the idea behind \emph{Pareto optimality}, which requires that there is no way to make someone better off without making anyone worse off. In the example above, there is no EFX allocation that is Pareto optimal.

One may still hope that this issue is pathological, stemming from unrealistically allowing agents to have no value for some goods. In general, if we require Pareto optimality but have several tiny goods valued by only one agent, the ``up to any good'' power of EFX is completely neutralized.\footnote{In fact, the original definition of EFX \citep{CaragiannisKMPSW16,CaragiannisKMPSW19} only considers removing goods that the envious agent values positively, and under that definition, giving $a$ and $b$ to agent 1 and $c$ to agent 2 is both EFX and Pareto optimal. The example uses the stronger definition of Plaut and Roughgarden~\citep{PlautR18,PlautR20}, which considers removing any good. For positive valuations, which are the focus of this paper, the two definitions coincide.} When every agent values every good (e.g., replacing each 0 above with an $\eps$), this conflict disappears, and we can open the door to positive results. Plaut and Roughgarden~\citep{PlautR18,PlautR20} showed that under \emph{nonzero marginal utility}, Pareto optimal EFX allocations exist for two agents with additive valuations, and for any number of agents with identical valuations. They also exist for \emph{bivalued} preferences where every value is one of two fixed positive numbers \citep{AmanatidisBFHV21,GargM23}, for \emph{lexicographic} preferences \citep{HosseiniSVX21}, and when there are only two \emph{types} of goods \citep{DavidiukDIS26}.

However, we show that Pareto optimal EFX allocations need not exist even when every agent values every good, answering a question raised by Plaut and Roughgarden~\citep{PlautR18,PlautR20}.\footnote{The same result can also be obtained with instances in which every agent has distinct values for distinct bundles. This rules out an approach to building a complete EFX allocation by maintaining Pareto optimality at every step, as suggested by Chaudhury~\citep[p.~154]{RayChaudhury21Thesis}. See \Cref{sec:ndg} for details.}

\begin{theorem}\label{thm:main}
The following instance with four agents and eight goods has no allocation that is both EFX and Pareto optimal.
\begin{center}\normalfont
\setlength{\tabcolsep}{11pt}\renewcommand{\arraystretch}{1.13}
\begin{tabular}{c|rr|rrrr|rr}
 & $x$&$y$&$a$&$b$&$c$&$d$&$p$&$q$\\\hline
1&7&6&3&3&1&1&2&1\\
2&6&7&1&3&3&1&1&2\\
3&7&6&1&1&3&3&2&1\\
4&6&7&3&1&1&3&1&2
\end{tabular}
\end{center}
\end{theorem}

With only $4^8 = 65,536$ allocations, \Cref{thm:main} is easy to verify programmatically. Our primary goal is to give an intuitive explanation of how this example exposes tensions between EFX and Pareto optimality. We then give a short proof which leverages the symmetry of the instance to reduce to only three cases. 

One might hope that the incompatibility in \cref{thm:main} is fragile, and that a slight relaxation of EFX is compatible with Pareto optimality. A relaxation which has been a source of interesting positive results is \emph{approximate EFX}. For $\alpha \in (0,1]$, an allocation is \emph{$\alpha$-EFX} if, whenever an agent looks at another agent's bundle with any single good removed, she values her own bundle at least an $\alpha$ fraction as much. $\alpha$-EFX allocations are known to exist for $\alpha = 1/2$ \citep{PlautR18,PlautR20} and $\alpha = \frac{\sqrt{5} - 1}{2} \approx 0.618$ \citep{AmanatidisMN20}, and for $\alpha = 2/3$ in several special cases \citep{AmanatidisFS24}. The proof of \cref{thm:main} in fact shows that the instance has no Pareto optimal allocation that is $\alpha$-EFX for any $\alpha > 6/7$, but we can show that the incompatibility persists all the way down to $2/3$.

\begin{theorem}\label{thm:approx}
For every $\alpha > 2/3$, there is an instance with four agents, nine goods, and positive additive valuations, in which no allocation is both $\alpha$-EFX and Pareto optimal.
\end{theorem}

In contrast, we do not yet know whether $\alpha$-EFX and Pareto optimal allocations always exist for \emph{any} constant $\alpha > 0$ (with positive valuations\footnote{With zero values, the answer is no. In \Cref{tab:efx-po-zero},
if each agent's value for $a$ is changed from $2$ to a large number $K$, then any Pareto optimal allocation is at most $\frac{1}{K}$-EFX.}), and we leave this as a challenging open problem (\Cref{op:eps-efx}).

\subsection{Related work}

\paragraph{Existence of EFX.} EFX allocations exist for two agents with general monotone valuations and for any number of agents with identical valuations \citep{PlautR18,PlautR20}, and for three agents with additive valuations \citep{ChaudhuryGM20,ChaudhuryGM24}. See the survey~\citep{AmanatidisABFLMVW23} for many other special cases. For four or more agents with additive valuations, the question remains open, although EFX allocations are known to exist for four agents with at most nine goods \citep{AlkassarFM26}. Very recently, EFX allocations were shown to not always exist for submodular valuations \citep{AkramiMMSW26,MackenzieS26EFX}, and for additive valuations in the \emph{chores} setting in which items have \emph{costs} (negative utilities) \citep{HeT26}. 
For subadditive valuations, Mackenzie and Suzuki~\citep{MackenzieS26EFX} also rule out $\alpha$-EFX for $\alpha > 2^{-1/6} \approx 0.89$.

\paragraph{Fairness and efficiency.}
As briefly mentioned earlier, Pareto optimal EFX allocations have been shown to exist in many special cases. Under nonzero marginal utility, Plaut and Roughgarden~\citep{PlautR18,PlautR20} show that the \emph{leximin} solution\footnote{Leximin maximizes the smallest utility, then the second-smallest, and so on. For the two-agent result, valuations are first normalized so that both agents have the same total value for goods.} gives an EFX and Pareto optimal allocation for two agents with additive valuations, and for any number of agents with identical valuations. Such allocations also exist for lexicographic preferences, where an agent prefers any good to all her less-preferred goods combined \citep{HosseiniSVX21}, and for positive additive valuations with only two types of goods, where each agent values all goods of a given type equally \citep{DavidiukDIS26}. \citep{AmanatidisBFHV21} showed that the allocation maximizing \emph{Nash welfare}\footnote{Nash welfare is the geometric mean of the agents' utilities.} is EFX and Pareto optimal in bivalued instances where every value is one of two fixed positive numbers. Thus, any counterexample for goods with positive additive valuations must use at least three distinct values. The instances in \Cref{thm:main,thm:approx} both use five.

On the negative side, deciding whether a Pareto optimal EFX allocation exists for additive valuations is NP-hard when zero values are allowed, even when there are at most three more goods than agents \citep{LimNT26}. For chores, He and Tao~\citep{HeT26} recently showed that EFX and Pareto optimality can be incompatible even when all costs are positive and take only two values. With three cost values, they show that EFX allocations alone can fail to exist. 

Pareto optimality has also been studied alongside closely related fairness objectives. For additive goods, maximizing Nash welfare always gives a Pareto optimal EF1 allocation \citep{CaragiannisKMPSW16,CaragiannisKMPSW19}, and allocations satisfying both properties can also be found in pseudopolynomial time~\citep{BarmanKV18}. However, when valuations are submodular rather than additive, Pareto optimal EF1 allocations may not exist even for two agents \citep{MackenzieS26b,ChandramouleeswaranN26}. For the related notion of equitability up to any good (EQX)\footnote{EQX requires $v_i(A_i) \ge v_j(A_j\setminus\{g\})$ for all agents $i,j$ and goods $g\in A_j$ with $v_j(g)>0$, comparing only agents' own values for their bundles.}, the leximin solution always gives a Pareto optimal EQX allocation for positive additive valuations \citep{FreemanSVX19}.

\section{Preliminaries}\label{sec:prelim}

A fair division instance consists of a set $N = \{1, \dots, n\}$ of agents, a set $M$ of indivisible goods, and valuation functions $v_i: 2^M\to \mathbb{R}_{\geq 0}$ for each agent $i \in N$. In the most general case, these valuation functions are arbitrary monotone functions, but we only need to consider \emph{additive} valuation functions which satisfy $v_i(S) = \sum_{g \in S} v_i(g)$ for every $S \subseteq M$. In this case, the valuation functions can be represented succinctly as a table giving the value that each agent has for each good. The instances we consider only use positive integers for values.

An \emph{allocation} $A = (A_1, \dots, A_n)$ is a partition of $M$ into $n$ sets (called \emph{bundles}), where agent $i$ receives $A_i$.  For brevity, we write bundles as strings and allocations as tuples, so that $(x, y, cdp, abq)$ is the allocation in which agents 1, 2, 3, and 4 receive $\{x\}$, $\{y\}$, $\{c, d, p\}$, and $\{a, b, q\}$, respectively.

\begin{definition}[Approximate EFX]
For $\alpha \in (0, 1]$, an allocation $A$ is \emph{$\alpha$-EFX} if for all agents $i \ne j$ and all goods $g \in A_j$,
\[
v_i(A_i) \ \ge\ \alpha \cdot v_i(A_j \setminus \{g\}).
\]
An allocation is \emph{EFX} if it is $1$-EFX.
\end{definition}

Since our valuations are additive, the tightest constraint comes from the good in $A_j$ that $i$ values the least. The constraint holds automatically if $A_j$ consists of a single good. We say that a good is a \emph{singleton} if it is the only good in its bundle.

\begin{definition}[Pareto optimality]
An allocation $B$ \emph{Pareto dominates} (or simply \emph{dominates}) an allocation $A$ if $v_i(B_i)
 \geq v_i(A_i)$ for all agents $i$, and at least one inequality is strict. An allocation is \emph{Pareto optimal} (also called \emph{Pareto efficient}) if no other allocation Pareto dominates it.
\end{definition}

In reasoning about Pareto optimality, it is convenient to identify an allocation $A$ with its \emph{utility profile} $\uvec(A) = (v_1(A_1), \dots, v_n(A_n))$. If an allocation dominates $A$, then it also dominates every allocation with the same utility profile, so this view is helpful for reducing casework.

\section{EFX and Pareto Optimality are Incompatible}

For convenience, we repeat the instance from \Cref{thm:main} below.

\begin{center}\normalfont
\setlength{\tabcolsep}{11pt}\renewcommand{\arraystretch}{1.13}
\begin{tabular}{c|rr|rrrr|rr}
 & $x$&$y$&$a$&$b$&$c$&$d$&$p$&$q$\\\hline
1&7&6&3&3&1&1&2&1\\
2&6&7&1&3&3&1&1&2\\
3&7&6&1&1&3&3&2&1\\
4&6&7&3&1&1&3&1&2
\end{tabular}
\end{center}

Before diving into the proof of the theorem, we will walk through some of the structure in the instance, and the tensions that force EFX allocations to be inefficient.\footnote{Readers may enjoy attempting to find a Pareto optimal EFX allocation in this example to get a sense for why they fail. Reasoning through it feels similar to solving a logic puzzle.}

\paragraph{Structure and symmetry.} The goods have three types: $x$ and $y$ are \emph{large}, $a, b, c, d$ are \emph{medium}, and $p$ and $q$ are \emph{small}. One may notice that agents' values for goods are very symmetric. The even (or odd) agents have the same preferences for large and small goods, and opposite preferences for medium goods. Any even agent and any odd agent have opposite preferences for the large and small goods, and have the same preference for two medium goods (one valued at 3, and one valued at 1), and opposite preferences for the other medium goods.

This symmetry is best characterized by a square (formally, the dihedral group of order $8$).
We can identify the agents with the vertices of a square in cyclic order, and associate the medium goods with the sides, and the pairs $\{x, p\}$ and $\{y, q\}$ with the two diagonals, as in \cref{fig:square-instance}. Think of each good as having a base value ($6$ for large goods and $1$ for the others), and if an agent is incident to a good's edge, then the agent's value for the good has an added bonus ($2$ for medium goods, and $1$ for the others).

\begin{figure}[htbp]
\centering
\begin{tikzpicture}[x=1cm,y=1cm,
  every node/.style={text=black,font=\normalfont\normalsize},
  square edge/.style={draw=black,line width=.5pt},
  agent/.style={circle,fill=white,line width=1.2pt,
    minimum size=6.5mm,inner sep=0pt,outer sep=0pt},
  good/.style={fill=white,inner sep=3pt},
  rotation/.style={draw=black,line width=.5pt,-{Stealth[length=1.6mm]}},
  reflection/.style={draw=black,line width=.5pt,dash pattern=on 2pt off 2pt},
  element/.style={font=\normalfont\small,inner sep=1pt}]

  \begin{scope}[shift={(0,-1.35)}]
    \DihedralSquare{1.68}{0}{1}{1}
  \end{scope}

  \begin{scope}[shift={(4.5,0)}]
    \foreach \k in {0,1,2,3}{%
      \coordinate (R\k) at (2.45*\k,0);
      \coordinate (S\k) at (2.45*\k,-2.7);
    }
    \foreach \k/\next in {0/1,1/2,2/3}{%
      \draw[rotation] ([xshift=.87cm]R\k)--([xshift=-.87cm]R\next);
      \draw[rotation] ([xshift=-.87cm]S\next)--([xshift=.87cm]S\k);
    }
    \draw[rotation,rounded corners=5pt]
      ([xshift=.87cm]R3)--(8.5,0)--(8.5,1.13)
      --(-1.15,1.13)--(-1.15,0)--([xshift=-.87cm]R0);
    \draw[rotation,rounded corners=5pt]
      ([xshift=-.87cm]S0)--(-1.15,-2.7)--(-1.15,-4.05)
      --(8.5,-4.05)--(8.5,-2.7)--([xshift=.87cm]S3);

    \foreach \k in {0,1,2,3}{%
      \draw[reflection] ([yshift=-1.15cm]R\k)--([yshift=.86cm]S\k);
    }

    \foreach \k/\toplabel/\bottomlabel in
      {0/{\mathrm{id}}/s,1/r/{sr},2/{r^2}/{sr^2},3/{r^3}/{sr^3}}{%
      \begin{scope}[shift={(R\k)},
        square edge/.style={draw=black,line width=.35pt},
        agent/.style={circle,fill=white,line width=.85pt,
          minimum size=3.5mm,inner sep=0pt,outer sep=0pt,
          font=\normalfont\scriptsize},
        good/.style={fill=white,inner sep=.5pt,
          font=\normalfont\fontsize{7}{8}\selectfont}]
        \DihedralSquare{.60}{\k}{1}{0}
      \end{scope}
      \begin{scope}[shift={(S\k)},
        square edge/.style={draw=black,line width=.35pt},
        agent/.style={circle,fill=white,line width=.85pt,
          minimum size=3.5mm,inner sep=0pt,outer sep=0pt,
          font=\normalfont\scriptsize},
        good/.style={fill=white,inner sep=.5pt,
          font=\normalfont\fontsize{7}{8}\selectfont}]
        \DihedralSquare{.60}{\k}{-1}{0}
      \end{scope}
      \node[element] at ([yshift=-.98cm]R\k) {$\toplabel$};
      \node[element] at ([yshift=-.98cm]S\k) {$\bottomlabel$};
    }
  \end{scope}

\end{tikzpicture}
\caption{The square representation and its dihedral symmetries.
Left: each good labels the edge joining the two agents who value it most.
Right: a Cayley graph of the dihedral group of order~8. Solid arrows
apply a $90^\circ$ clockwise rotation $r$.
Dashed lines apply reflection $s$ across a vertical axis.}
\label{fig:square-instance}
\end{figure}

With the square in mind, we call agents \emph{adjacent} if they are incident to a common side, and \emph{opposite} if they are incident to a common diagonal. 
Rotating and reflecting
the square allows us to relabel agents and goods into canonical configurations without loss of generality, and we will do so frequently. (See \Cref{fig:square-instance} for a reminder of which transformations are possible.)

The symmetry is convenient for reducing casework, but the square is also a helpful explanation for the conflicts in the instance when viewed as a graph. The edges show how pairs of agents compete for the same goods, and intuitively, for an allocation to avoid inefficiency, it ought to give each good to an agent incident to it. 

\paragraph{Inefficiencies caused by EFX.} There are a few primary effects which lead to EFX allocations being Pareto inefficient. The first is that EFX allocations must allocate the large goods as singletons, much like in \Cref{tab:efx-po-zero}.
The intuition is that if a large good shares a bundle with anything else, then every other agent must value her bundle as much as the large good alone. Since the large goods are worth 6 or 7, while the six other goods are worth only at most 11 to each agent, this constraint turns out to be too restrictive. 

The singleton constraint inherently causes a loss in the total value (compared to the maximum possible by allocating each good to an incident agent).
If such a loss were avoidable, the two agents that get large goods (called \emph{lucky}) should be adjacent, which makes the other two agents (called \emph{unlucky}) adjacent as well. But if adjacent agents split the medium goods, one of these goods goes to an agent that has value 1 for it, which is a loss of 2 units of value. Similarly, if the lucky
agents are opposite, then we lose 1 unit in allocating the large goods, and another unit in allocating the small goods.

\paragraph{Exploiting inefficiencies with trades.}

Of course, lost value alone does not make an allocation Pareto dominated. The key is that some lost value can always be recovered through carefully choreographed trades (that leave no one worse off). To see how these trades might work, notice that a large good is worth about as much as two medium goods: an agent who holds a large good worth $7$ to her would gladly trade it for her favorite medium goods and small good, worth $3 + 3 + 2 = 8$, and  an agent who holds her two favorite medium goods would be willing to give them up for a large good worth $6$. Similarly, a favorite medium good is worth as much as two small goods (or one small and one less preferred medium good).

\Cref{fig:efx-pareto-improvement-1} gives an illustrative example of these trades in action. A natural EFX allocation is $A = (x, y, cdp, abq)$, giving the large goods to agents 1 and 2, and each other good to agent 3 or 4 depending on who values it more (breaking the ties for $b$ and $d$ in opposite directions to be fair).
To reduce the unrealized value, we want to allocate the medium goods to opposite agents. We can do this by trading agent 3's medium goods for $y$, and trading agent 4's $b$ and $q$ for $d$. This allows agent $2$ to get $bcq$ as compensation for $y$, which is a net gain. The resulting allocation is $B=(x,bcq,yp,ad)$, which Pareto dominates $A$, but is not EFX since the large good $y$ is not allocated alone (in particular, agent 4 prefers $y$ to her bundle $ad$).

For the sake of brevity, we will just provide the Pareto dominating allocations instead of explaining the trades in detail in the actual proofs, but one can find the exact same kinds of moves taking place in \Cref{fig:efx-pareto-improvements}.

\begin{figure}[htbp]
\centering
\setlength{\tabcolsep}{6pt}
\renewcommand{\arraystretch}{1.14}
\begin{tabular}{c|cc|cccc|cc!{\vrule width 0.8pt}cc}
&$x$&$y$&$a$&$b$&$c$&$d$&$p$&$q$
&$\uvec(A)$&$\uvec(B)$\\\hline
1&\keepval{7}&\plainval{6}&\plainval{3}&\plainval{3}&\plainval{1}&\plainval{1}&\plainval{2}&\plainval{1}&7&7\\
2&\plainval{6}&\oldval{7}&\plainval{1}&\newval{3}&\newval{3}&\plainval{1}&\plainval{1}&\newval{2}&7&\textbf{8}\\
3&\plainval{7}&\newval{6}&\plainval{1}&\plainval{1}&\oldval{3}&\oldval{3}&\keepval{2}&\plainval{1}&8&8\\
4&\plainval{6}&\plainval{7}&\keepval{3}&\oldval{1}&\plainval{1}&\newval{3}&\plainval{1}&\oldval{2}&6&6\\
\end{tabular}
\caption{The EFX allocation $A=(x,y,cdp,abq)$ (boxed) is
Pareto dominated by $B=(x,bcq,yp,ad)$ (green): agent~2 gains
one unit of utility, while every other agent's utility is unchanged.}
\label{fig:efx-pareto-improvement-1}
\end{figure}

We now prove \Cref{thm:main} formally.

\begin{proof}[Proof of \Cref{thm:main}.]

We split the proof into three steps. First, we show that in every EFX allocation, $x$ and $y$ are allocated as singletons. Then, we argue that up to symmetry, there are only three \emph{maximal} utility profiles for EFX allocations (that are not dominated by another EFX allocation). Finally, we give an allocation that dominates each of these profiles.  

\paragraph{Step 1: EFX isolates large goods.} Fix some EFX allocation. If one agent has both $x$ and $y$, then every other agent must get value at least $7$ by EFX. However, the total value for the remaining goods is at most $3\cdot 4 + 2\cdot 2 = 16 < 7\cdot 3$, so $x$ and $y$ must have different owners. Call these agents \emph{lucky} and the other two \emph{unlucky}. 

If both lucky agents have a second good, then the unlucky agents must have value at least $7$ by EFX. But they have at most four goods between them, worth at most $12 < 7 \cdot 2$ in total. Thus, at least one of the large goods is a singleton. 

Suppose that $x$ is a singleton but $y$ is not. By EFX, every agent values her bundle at least as much as $y$. In particular the agent who gets $x$ must prefer $x$ to $y$, so she is odd, say agent $1$ after reflection.
The unlucky agents split at most five goods, so one has at most two. This agent can only value her bundle as much as $y$ if she is odd (agent $3$) and has her favorite medium goods ($cd$). But then the other unlucky agent is even, gets three of $abpq$, and has value at least $7$ (her value for $y$), which is impossible.

Thus, both $x$ and $y$ are allocated as singletons. One more helpful observation is that the unlucky agents must each get at least two of the remaining goods. If an unlucky agent gets at most one, her value is at most $3$, but her value for the other unlucky agent's bundle after removing a small good is at least $5$.

\paragraph{Step 2: Three maximal utility profiles remain.} Next, we analyze the possible maximal utility profiles, up to symmetry.

If both large goods are worth 6 to their owners, we can swap them and get a Pareto improvement. Thus, after rotation we can assume the lucky agents are 1 and 2 or 1 and 3. These cases correspond to utility profiles of the form $(7,7,s,t)$ and $(7, s, 6, t)$ respectively. In either case we can swap the unlucky agents with a reflection, so we only need to consider $s \geq t$.

\emph{Case 1: $(7,7,s,t)$}. The small and medium goods have total value at most $3 + 1 + 3 + 3 + 2 + 2 = 14$ to agents 3 and 4, so $s + t\leq 14$. We cannot get $s = t = 7$, since it would require each of the six goods to go to the agent that prefers it (breaking ties for $b$ and $d$), and this can only result in $(s, t) = (8, 6)$ or $(9, 5)$. Since agent 3 gets at most four goods, $s \leq 9$, so every profile in this case is dominated component-wise by either $(7,7,8,6)$ or $(7,7,9,5)$.

\emph{Case 2: $(7,s,6,t)$}. The small and medium goods have total value at most $3\cdot 4 + 1 + 2 = 15$ to agents 2 and 4, so $s + t\leq 15$. Since $s\leq 9$ for the same reason as before, $(s,t) \leq (9, 6)$ or $(8, 7)$. This leaves maximal profiles $(7, 9, 6, 6)$ and $(7, 8, 6, 7)$.

Altogether, we have four maximal profiles: 
$$(7,7,8,6), \qquad (7,7,9,5), \qquad (7, 9, 6, 6), \qquad  (7, 8, 6, 7).$$
The first and last are actually the same after a rotation, so we only need to dominate the first three profiles.

\paragraph{Step 3: Dominate the three profiles.} It remains to show that the profiles  $(7,7,8,6), (7,7,9,5),$ and  $(7, 9, 6, 6)$ have Pareto improvements. The first is shown in \Cref{fig:efx-pareto-improvement-1}, and the others are shown in \Cref{fig:efx-pareto-improvements}. (We include an example EFX allocation in each for readers who are interested in tracing the trades.)

\begin{figure}[htbp]
\centering
\small
\setlength{\tabcolsep}{4pt}
\renewcommand{\arraystretch}{1.2}

\begin{minipage}[t]{0.49\linewidth}
\centering
\begin{tabular}{c|cc|cccc|cc!{\vrule width 0.8pt}cc}
&$x$&$y$&$a$&$b$&$c$&$d$&$p$&$q$
&$\uvec(A)$&$\uvec(B)$\\\hline
1&\oldval{7}&\plainval{6}&\newval{3}&\newval{3}&\plainval{1}&\plainval{1}&\newval{2}&\plainval{1}&7&\textbf{8}\\
2&\plainval{6}&\keepval{7}&\plainval{1}&\plainval{3}&\plainval{3}&\plainval{1}&\plainval{1}&\plainval{2}&7&7\\
3&\newval{7}&\plainval{6}&\plainval{1}&\oldval{1}&\keepval{3}&\oldval{3}&\oldval{2}&\plainval{1}&9&\textbf{10}\\
4&\plainval{6}&\plainval{7}&\oldval{3}&\plainval{1}&\plainval{1}&\newval{3}&\plainval{1}&\keepval{2}&5&5\\
\end{tabular}
\end{minipage}\hfill%
\begin{minipage}[t]{0.49\linewidth}
\centering
\begin{tabular}{c|cc|cccc|cc!{\vrule width 0.8pt}cc}
&$x$&$y$&$a$&$b$&$c$&$d$&$p$&$q$
&$\uvec(A)$&$\uvec(B)$\\\hline
1&\oldval{7}&\plainval{6}&\newval{3}&\newval{3}&\plainval{1}&\plainval{1}&\newval{2}&\plainval{1}&7&\textbf{8}\\
2&\plainval{6}&\newval{7}&\plainval{1}&\oldval{3}&\oldval{3}&\plainval{1}&\oldval{1}&\keepval{2}&9&9\\
3&\plainval{7}&\oldval{6}&\plainval{1}&\plainval{1}&\newval{3}&\newval{3}&\plainval{2}&\plainval{1}&6&6\\
4&\newval{6}&\plainval{7}&\oldval{3}&\plainval{1}&\plainval{1}&\oldval{3}&\plainval{1}&\plainval{2}&6&6\\
\end{tabular}
\end{minipage}

\caption{Two EFX allocations $A$ (boxed) and Pareto improvements
$B$ (green). Left: $A=(x,y,bcdp,aq)$ and $B=(abp,y,xc,dq)$;
agents~1 and~3 each gain one unit of utility.
Right: $A=(x,bcpq,y,ad)$ and $B=(abp,yq,cd,x)$;
agent~1 gains one unit. All other utilities are unchanged.}
\label{fig:efx-pareto-improvements}
\end{figure}
\end{proof}

\section{Approximate EFX}

The instance in \Cref{thm:main} is the simplest example we are aware of with positive values and no Pareto optimal EFX allocation, but the idea behind it can be strengthened to rule out Pareto optimal \emph{approximate} EFX allocations as well. It is not hard to see that \Cref{thm:main} already shows that $\alpha$-EFX and Pareto optimality are incompatible for $\alpha > 6/7$, but with some modifications we can extend the impossibility down to $2/3$.

\begin{approxtheoremrestated}
For every integer $K \geq 6$, the following instance has no Pareto optimal $\alpha$-EFX allocation for $\alpha > \frac23 + \frac1K$.
\begin{center}\normalfont
\setlength{\tabcolsep}{7pt}\renewcommand{\arraystretch}{1.12}
\begin{tabular}{c|cc|ccccc|cc}
 &$x$&$y$&$a$&$b$&$c$&$d$&$e$&$p$&$q$\\\hline
1&$3K$&$2K+2$&$K$&$K$&1&1&$K$&2&1\\
2&$2K+2$&$3K$&1&$K$&$K$&1&$K$&1&2\\
3&$3K$&$2K+2$&1&1&$K$&$K$&$K$&2&1\\
4&$2K+2$&$3K$&$K$&1&1&$K$&$K$&1&2
\end{tabular}
\end{center}
\end{approxtheoremrestated}

The proof follows the same template as before: the large goods are too valuable to be allocated with other goods, forcing a loss in value that can be recovered with careful trades without leaving anyone worse off. The most significant difference in the instance is the additional good $e$ which has the same value for every agent. This additional medium good allows agents to trade a large good for their favorite medium goods without losing value (akin to the last proof).

If we follow the proof step by step, we find that $\alpha$-EFX allocations can have five different maximal utility profiles, but these are actually dominated by only two different allocations (up to symmetry). With this in mind, we can streamline the proof by introducing the dominating allocations earlier, at the point where they can cover the cases at hand.

\begin{proof}[Proof of \Cref{thm:approx}]
Fix an $\alpha$-EFX allocation $A$ where $\alpha > \frac23 + \frac1K$. Thus, 
$v_i(A_i) \geq \alpha \cdot v_i(A_j\setminus \{g\})$
for all $i \neq j$ and $g\in A_j$. To show that certain choices violate $\alpha$-EFX, we will use the following three inequalities:
$$ 2K+3 < \alpha\cdot3K,\qquad K+5<\alpha(2K+2),\qquad K+4<\alpha\cdot2K.$$
The last two hold when $K$ is sufficiently large, specifically at least $6$.

\paragraph{Step 1: $\alpha$-EFX isolates large goods.} If one agent has both $x$ and $y$, each other agent has utility greater than $2K+3$ by $\alpha$-EFX. Since every nonlarge good is worth at most $K$, each needs at least three goods, but there are only seven to share. Thus $x,y$ have different owners. Call them \emph{lucky} and the other two agents \emph{unlucky}.

If both lucky agents have a second good, each unlucky agent again needs utility greater than $2K+3$ by $\alpha$-EFX. Each therefore needs at least three goods, but at most five remain. Thus, at least one large good is a singleton.

Suppose, by symmetry, that $x$ is a singleton but $y$ is not. The core issue is that by $\alpha$-EFX, every agent values her bundle at least $\alpha$ times her value of $y$. This is greater than $2K + 3$ for even agents and $K + 5$ for odd agents. Each unlucky agent must get at least two favorite medium goods. Otherwise, her value is at most $K + 5$. Since the lucky agents take the large goods and one more, the unlucky agents share at most six goods. Each gets at least two, so neither can get more than four. With at most two favorite medium goods, such a bundle is worth at most $2K + 3$.

But the owner of $x$ must be odd, since it is a singleton and even agents value it at $2K + 2$. So at least one of the unlucky agents is even. To get value greater than $2K + 3$, this unlucky agent must get all three of her favorite medium goods. But then the only way for the other unlucky agent to get two favorite medium goods is if she too is even, in which case she actually must get all three. Thus, both unlucky agents would have to get $e$. This contradiction implies that $x$ and $y$ must both be singletons.

One more useful observation is that each unlucky agent receives at least two goods. Otherwise, her utility is at most $K$, while the other unlucky bundle, after deleting one of $p, q$, is worth at least $2K + 3$ to her, violating $\alpha$-EFX.

\paragraph{Step 2: Finding dominating allocations.} Suppose our fixed $\alpha$-EFX allocation $A$ is also Pareto optimal. If both large goods are worth $2K + 2$ to their owners, we can swap them and get a Pareto improvement. Thus, after rotation we can assume that the lucky agents are 1 and 2 or 1 and 3.

\emph{Case 1: Adjacent lucky agents $1,2$.} Take $A_1 = x$, $A_2 = y$ and reflect so that the utility profile $(3K, 3K, s, t)$ satisfies $s \geq t$. We claim that $s \leq 2K + 4$, $t \leq 2K + 2$, and $s + t \leq 4K + 5$. This implies that $A$ is dominated by $B=(x,bce,yp,adq)$, with utility profile $\uvec(B)=(3K,3K,2K+4,2K+2)$. Every agent gets at least as much utility, and the total is at least one more.

\begin{center}
\begin{minipage}{\linewidth}\centering
\setlength{\tabcolsep}{4.5pt}\renewcommand{\arraystretch}{1.17}
\begin{tabular}{c|cc|ccccc|cc!{\vrule width 0.8pt}cc}
&$x$&$y$&$a$&$b$&$c$&$d$&$e$&$p$&$q$
&$\uvec(A)$&$\uvec(B)$\\\hline
1&\keepval{3K}&\plainval{2K+2}
 &\plainval{K}&\plainval{K}&\plainval{1}&\plainval{1}&\plainval{K}
 &\plainval{2}&\plainval{1}&$3K$&$3K$\\
2&\plainval{2K+2}&\oldval{3K}
 &\plainval{1}&\newval{K}&\newval{K}&\plainval{1}&\newval{K}
 &\plainval{1}&\plainval{2}&$3K$&$3K$\\
3&\plainval{3K}&\newval{2K+2}
 &\plainval{1}&\plainval{1}&\plainval{K}&\plainval{K}&\plainval{K}
 &\newval{2}&\plainval{1}&$s$&$2K+4$\\
4&\plainval{2K+2}&\plainval{3K}
 &\newval{K}&\plainval{1}&\plainval{1}&\newval{K}&\plainval{K}
 &\plainval{1}&\newval{2}&$t$&$2K+2$\\
\end{tabular}
\end{minipage}
\end{center}

It is easiest to show the three inequalities in reverse order. We know that $s + t \leq 4K + 5$, since this is the maximum total utility of the nonlarge goods for agents 3 and 4. Then $t \leq 2K + 2$ follows from $t \leq s$. 

Now suppose towards a contradiction that $s > 2K + 4$. Then agent 3 must get all three of her favorite medium goods, since a bundle with only two favorite medium goods and at most three others is worth at most $2K + 4$ to her.  Agent 4 can receive only goods from $abpq$, worth at most $K + 4$ to her, while she values $A_3 \setminus \{c\}$ at least $2K$. Since $K + 4 < \alpha \cdot 2K$, this violates $\alpha$-EFX. Thus, $s \leq 2K + 4$ as claimed.

\emph{Case 2: Opposite lucky agents $1,3$.} Take $A_1=x$, $A_3=y$, and reflect so that agent~2 holds $e$. Write $\uvec(A)=(3K,s,2K+2,t)$ (here we do not need $s\ge t$). We will show that $A$ is dominated by one of two allocations, depending on whether $s \leq 3K$ or $s \geq 3K + 1$.

First, note that $s + t \leq 5K + 3$, since this is the maximum total utility of the nonlarge goods for agents 2 and 4. Since both unlucky agents get at least two goods, agent 4 gets at most five. Without $e$, these are worth at most $K + K + 2 + 1 + 1 = 2K + 4$ to her. Thus, if $s \leq 3K$, then $A$ is dominated by $B=(abe,y,cdp,xq)$, with utility profile $\uvec(B)=(3K,3K,2K+2,2K+4)$. Every agent gets at least as much utility, and the total is at least one larger. (Note that this $B$ is the same as the last one, after reflection.)

\begin{center}
\begin{minipage}{\linewidth}\centering
\setlength{\tabcolsep}{4.5pt}\renewcommand{\arraystretch}{1.17}
\begin{tabular}{c|cc|ccccc|cc!{\vrule width 0.8pt}cc}
&$x$&$y$&$a$&$b$&$c$&$d$&$e$&$p$&$q$
&$\uvec(A)$&$\uvec(B)$\\\hline
1&\oldval{3K}&\plainval{2K+2}
 &\newval{K}&\newval{K}&\plainval{1}&\plainval{1}&\newval{K}
 &\plainval{2}&\plainval{1}&$3K$&$3K$\\
2&\plainval{2K+2}&\newval{3K}
 &\plainval{1}&\plainval{K}&\plainval{K}&\plainval{1}&\oldval{K}
 &\plainval{1}&\plainval{2}&$s$&$3K$\\
3&\plainval{3K}&\oldval{2K+2}
 &\plainval{1}&\plainval{1}&\newval{K}&\newval{K}&\plainval{K}
 &\newval{2}&\plainval{1}&$2K+2$&$2K+2$\\
4&\newval{2K+2}&\plainval{3K}
 &\plainval{K}&\plainval{1}&\plainval{1}&\plainval{K}&\plainval{K}
 &\plainval{1}&\newval{2}&$t$&$2K+4$\\
\end{tabular}
\end{minipage}
\end{center}

Next, suppose that $s \geq 3K + 1$. Then $s + t \leq 5K + 3$ gives $t \leq 2K + 2$. We also have $s \leq 3K + 5$, since this is the total value of the nonlarge goods. Since $K \geq 6$, we have $s < 4K$, which means that $A$ is dominated by the allocation $C = (x,yb,cep,adq)$, with utility profile $\uvec(C)=(3K,4K,2K+2,2K+2)$.

\begin{center}
\begin{minipage}{\linewidth}\centering
\setlength{\tabcolsep}{4.5pt}\renewcommand{\arraystretch}{1.17}
\begin{tabular}{c|cc|ccccc|cc!{\vrule width 0.8pt}cc}
&$x$&$y$&$a$&$b$&$c$&$d$&$e$&$p$&$q$
&$\uvec(A)$&$\uvec(C)$\\\hline
1&\keepval{3K}&\plainval{2K+2}
 &\plainval{K}&\plainval{K}&\plainval{1}&\plainval{1}&\plainval{K}
 &\plainval{2}&\plainval{1}&$3K$&$3K$\\
2&\plainval{2K+2}&\newval{3K}
 &\plainval{1}&\newval{K}&\plainval{K}&\plainval{1}&\oldval{K}
 &\plainval{1}&\plainval{2}&$s$&$4K$\\
3&\plainval{3K}&\oldval{2K+2}
 &\plainval{1}&\plainval{1}&\newval{K}&\plainval{K}&\newval{K}
 &\newval{2}&\plainval{1}&$2K+2$&$2K+2$\\
4&\plainval{2K+2}&\plainval{3K}
 &\newval{K}&\plainval{1}&\plainval{1}&\newval{K}&\plainval{K}
 &\plainval{1}&\newval{2}&$t$&$2K+2$\\
\end{tabular}
\end{minipage}
\end{center}

Thus, we have shown that every $\alpha$-EFX allocation is dominated, which concludes the proof. \qedhere

\end{proof}

\section{Discussion}

Our results add to a growing list of recent counterexamples in fair division, from pairwise maximin share fairness for additive valuations \citep{Aziz26,Golz26,BeiLLLT26} to EFX for general monotone and even submodular valuations \citep{AkramiMMSW26,MackenzieS26EFX}. While our examples look very simple, they may have been  outside the reach of a brute-force style approach since the search space is larger than one might expect.\footnote{For \Cref{thm:main}, there are $7^{32} \approx 1.1\times 10^{27}$ instances with four agents, eight goods, and integer values between $1$ and $7$.} 
Perhaps we should not be surprised if a similarly simple construction, even exploiting the same tensions, gives a counterexample to EFX with additive valuations.  

Even if EFX allocations do always exist for additive valuations, their counterintuitive behavior (exemplified by the various instances in this paper) makes it worth considering whether EFX is always the right notion.
Would we rather have a potentially unwieldy algorithm that guarantees EFX, or a simple objective like maximum Nash welfare, which guarantees both EF1 and Pareto optimality \citep{CaragiannisKMPSW16,CaragiannisKMPSW19}? 
It is not obvious whether we should be willing to give up economic efficiency for the stronger fairness guarantee.

With this in mind, we suggest two directions towards positive results balancing fairness and efficiency by relaxing our goals.

The first is to shift focus towards approximate EFX instead of exact EFX. Although approximate EFX is a little less natural at first glance, the algorithms achieving $1/2$-EFX \citep{PlautR18,PlautR20} and $0.618$-EFX \citep{AmanatidisMN20} are quite intuitive. Hoping that the same may still be true when adding Pareto optimality as a desideratum, we pose the following open problem, which we have not yet been able to resolve despite significant effort (human and otherwise). 

\begin{openproblem}\label{op:eps-efx}
Is there an absolute constant $\alpha>0$ such that every instance with positive additive valuations admits an allocation that is both $\alpha$-EFX and Pareto optimal?
\end{openproblem}

A second, more conceptual direction is to mitigate concerns by explaining why the kinds of instances explored in this paper appear to be exceedingly rare. While they may expose a real conflict between fairness and efficiency for indivisible goods, one could make an argument that these examples are still rather pathological. Existing work already gives models where, with sufficiently many goods, even exact envy-freeness and Pareto optimality coexist with high probability \citep{DickersonGKPS14,BaiG22}. But one very valuable good can put envy-freeness out of reach, without making the instance seem particularly pathological. Can we make a case for why EFX and Pareto optimality should almost always be compatible, with a natural average-case model that accounts for valuable goods?

\paragraph{Acknowledgments.} We thank Ariel Procaccia for helpful conversations on this question and feedback on a draft of this paper, and Amrit Daswaney for sharing an extension of \Cref{thm:main} to non-degenerate instances (\Cref{sec:ndg}).

\paragraph{AI Methodology.} The instances in \Cref{thm:main,thm:approx} were found using a custom agentic harness developed by the authors, which uses multiple models; the main idea was developed in cooperation between Claude Opus 5.5 and GPT-6 Astra. We then worked with these two models to find streamlined proofs with minimal casework and interpretable steps. The final exposition is written by the authors, with additional AI assistance in copy editing, checking correctness, finding relevant literature, and producing figures.

\bibliographystyle{alpha}
\bibliography{refs}

\newcommand{\etalchar}[1]{$^{#1}$}
\begin{thebibliography}{ABFR{\etalchar{+}}21}

\bibitem[AAB{\etalchar{+}}23]{AmanatidisABFLMVW23}
Georgios Amanatidis, Haris Aziz, Georgios Birmpas, Aris Filos-Ratsikas, Bo~Li, Herv{\'e} Moulin, Alexandros~A. Voudouris, and Xiaowei Wu.
\newblock Fair division of indivisible goods: Recent progress and open questions.
\newblock {\em Artificial Intelligence}, 322:103965, 2023.

\bibitem[ABFR{\etalchar{+}}21]{AmanatidisBFHV21}
Georgios Amanatidis, Georgios Birmpas, Aris Filos-Ratsikas, Alexandros Hollender, and Alexandros~A. Voudouris.
\newblock Maximum {Nash} welfare and other stories about {EFX}.
\newblock {\em Theoretical Computer Science}, 863:69--85, 2021.

\bibitem[AFM26]{AlkassarFM26}
Eyad Alkassar, Mahmoud Fouz, and Kurt Mehlhorn.
\newblock Complete {EFX} allocations exist for four additive agents and up to nine goods, 2026.
\newblock arXiv:2608.08590.

\bibitem[AFRS24]{AmanatidisFS24}
Georgios Amanatidis, Aris Filos-Ratsikas, and Alkmini Sgouritsa.
\newblock Pushing the frontier on approximate {EFX} allocations.
\newblock In {\em Proceedings of the 25th {ACM} Conference on Economics and Computation ({EC})}, pages 1268--1286. ACM, 2024.

\bibitem[AMM{\etalchar{+}}26]{AkramiMMSW26}
Hannaneh Akrami, Alexander Mayorov, Kurt Mehlhorn, Shreyas Srinivas, and Christoph Weidenbach.
\newblock A counterexample to {EFX} {$n \geq 3$} agents, {$m \geq n+5$} items, submodular valuations via {SAT}-solving, 2026.
\newblock arXiv:2604.18216.

\bibitem[AMN20]{AmanatidisMN20}
Georgios Amanatidis, Evangelos Markakis, and Apostolos Ntokos.
\newblock Multiple birds with one stone: Beating 1/2 for {EFX} and {GMMS} via envy cycle elimination.
\newblock {\em Theoretical Computer Science}, 841:94--109, 2020.

\bibitem[Azi26]{Aziz26}
Haris Aziz.
\newblock Pairwise maximin share allocations need not exist, 2026.
\newblock arXiv:2609.06282.

\bibitem[BG22]{BaiG22}
Yushi Bai and Paul G{\"o}lz.
\newblock Envy-free and {Pareto}-optimal allocations for agents with asymmetric random valuations.
\newblock In {\em Proceedings of the Thirty-First International Joint Conference on Artificial Intelligence ({IJCAI})}, pages 53--59, 2022.

\bibitem[BKV18]{BarmanKV18}
Siddharth Barman, Sanath~Kumar Krishnamurthy, and Rohit Vaish.
\newblock Finding fair and efficient allocations.
\newblock In {\em Proceedings of the 2018 {ACM} Conference on Economics and Computation ({EC})}, pages 557--574. ACM, 2018.

\bibitem[BLL{\etalchar{+}}26]{BeiLLLT26}
Xiaohui Bei, Zehan Lin, Shengxin Liu, Rong Luan, and Biaoshuai Tao.
\newblock Non-existence of {PMMS} allocations and a {$4/3$}-{PMMS} guarantee for additive chores, 2026.
\newblock arXiv:2609.10493.

\bibitem[Bud11]{Budish11}
Eric Budish.
\newblock The combinatorial assignment problem: Approximate competitive equilibrium from equal incomes.
\newblock {\em Journal of Political Economy}, 119(6):1061--1103, 2011.

\bibitem[CGM20]{ChaudhuryGM20}
Bhaskar~Ray Chaudhury, Jugal Garg, and Kurt Mehlhorn.
\newblock {EFX} exists for three agents.
\newblock In {\em Proceedings of the 21st {ACM} Conference on Economics and Computation ({EC})}, pages 1--19. ACM, 2020.

\bibitem[CGM24]{ChaudhuryGM24}
Bhaskar~Ray Chaudhury, Jugal Garg, and Kurt Mehlhorn.
\newblock {EFX} exists for three agents.
\newblock {\em Journal of the ACM}, 71(1):4:1--4:27, 2024.

\bibitem[Cha21]{RayChaudhury21Thesis}
Bhaskar~Ray Chaudhury.
\newblock {\em Finding Fair and Efficient Allocations}.
\newblock PhD thesis, Saarland University, 2021.

\bibitem[CKM{\etalchar{+}}16]{CaragiannisKMPSW16}
Ioannis Caragiannis, David Kurokawa, Herv{\'e} Moulin, Ariel~D. Procaccia, Nisarg Shah, and Junxing Wang.
\newblock The unreasonable fairness of maximum {Nash} welfare.
\newblock In {\em Proceedings of the 2016 {ACM} Conference on Economics and Computation ({EC})}, pages 305--322. ACM, 2016.

\bibitem[CKM{\etalchar{+}}19]{CaragiannisKMPSW19}
Ioannis Caragiannis, David Kurokawa, Herv{\'e} Moulin, Ariel~D. Procaccia, Nisarg Shah, and Junxing Wang.
\newblock The unreasonable fairness of maximum {Nash} welfare.
\newblock {\em ACM Transactions on Economics and Computation}, 7(3):12:1--12:32, 2019.

\bibitem[CN26]{ChandramouleeswaranN26}
Harish Chandramouleeswaran and Prajakta Nimbhorkar.
\newblock Nonexistence of simultaneously {EF1} and {Pareto} optimal allocations for submodular valuations, 2026.
\newblock arXiv:2607.18220.

\bibitem[DDIS26]{DavidiukDIS26}
Vladimir Davidiuk, Yuriy Dementiev, Artur Ignatiev, and Danil Sagunov.
\newblock {EFX} and {PO} allocation exists for two types of goods.
\newblock In {\em Proceedings of the {AAAI} Conference on Artificial Intelligence}, volume~40, pages 16795--16802, 2026.

\bibitem[DGK{\etalchar{+}}14]{DickersonGKPS14}
John~P. Dickerson, Jonathan~R. Goldman, Jeremy Karp, Ariel~D. Procaccia, and Tuomas Sandholm.
\newblock The computational rise and fall of fairness.
\newblock In {\em Proceedings of the Twenty-Eighth {AAAI} Conference on Artificial Intelligence}, pages 1405--1411. AAAI Press, 2014.

\bibitem[FSVX19]{FreemanSVX19}
Rupert Freeman, Sujoy Sikdar, Rohit Vaish, and Lirong Xia.
\newblock Equitable allocations of indivisible goods.
\newblock In {\em Proceedings of the Twenty-Eighth International Joint Conference on Artificial Intelligence ({IJCAI})}, pages 280--286. International Joint Conferences on Artificial Intelligence Organization, 2019.

\bibitem[GM23]{GargM23}
Jugal Garg and Aniket Murhekar.
\newblock Computing fair and efficient allocations with few utility values.
\newblock {\em Theoretical Computer Science}, 962:113932, 2023.

\bibitem[GMT14]{GourvesMT14}
Laurent Gourv{\`e}s, J{\'e}r{\^o}me Monnot, and Lydia Tlilane.
\newblock Near fairness in matroids.
\newblock In {\em Proceedings of the 21st European Conference on Artificial Intelligence ({ECAI})}, volume 263 of {\em Frontiers in Artificial Intelligence and Applications}, pages 393--398. IOS Press, 2014.

\bibitem[G{\"o}l26]{Golz26}
Paul G{\"o}lz.
\newblock {PMMS} allocations need not exist for 3 agents with additive valuations, 2026.
\newblock arXiv:2609.08954.

\bibitem[HSVX21]{HosseiniSVX21}
Hadi Hosseini, Sujoy Sikdar, Rohit Vaish, and Lirong Xia.
\newblock Fair and efficient allocations under lexicographic preferences.
\newblock In {\em Proceedings of the {AAAI} Conference on Artificial Intelligence}, volume~35, pages 5472--5480, 2021.

\bibitem[HT26]{HeT26}
Wentao He and Biaoshuai Tao.
\newblock {EFX} for additive chores: Nonexistence, {Pareto} incompatibility, and bi-valued existence, 2026.
\newblock arXiv:2606.08872.

\bibitem[LMMS04]{LiptonMMS04}
Richard~J. Lipton, Evangelos Markakis, Elchanan Mossel, and Amin Saberi.
\newblock On approximately fair allocations of indivisible goods.
\newblock In {\em Proceedings of the 5th {ACM} Conference on Electronic Commerce ({EC})}, pages 125--131. ACM, 2004.

\bibitem[LNT26]{LimNT26}
Eugene Lim, Tzeh~Yuan Neoh, and Nicholas Teh.
\newblock The cost of {EFX}: Generalized-mean welfare and complexity dichotomies with few surplus items, 2026.
\newblock arXiv:2601.12849.

\bibitem[MS26a]{MackenzieS26EFX}
Simon Mackenzie and Mashbat Suzuki.
\newblock Counterexamples to {EFX} for submodular and subadditive valuations, 2026.
\newblock arXiv:2605.06451.

\bibitem[MS26b]{MackenzieS26b}
Simon Mackenzie and Mashbat Suzuki.
\newblock When one good is not enough: {EF1} and {Pareto} optimality are not compatible for submodular valuations, 2026.
\newblock arXiv:2607.17811.

\bibitem[PR18]{PlautR18}
Benjamin Plaut and Tim Roughgarden.
\newblock Almost envy-freeness with general valuations.
\newblock In Artur Czumaj, editor, {\em Proceedings of the Twenty-Ninth Annual {ACM-SIAM} Symposium on Discrete Algorithms ({SODA})}, pages 2584--2603. SIAM, 2018.

\bibitem[PR20]{PlautR20}
Benjamin Plaut and Tim Roughgarden.
\newblock Almost envy-freeness with general valuations.
\newblock {\em SIAM Journal on Discrete Mathematics}, 34(2):1039--1068, 2020.

\bibitem[Pro20]{Procaccia20}
Ariel~D. Procaccia.
\newblock An answer to fair division's most enigmatic question: Technical perspective.
\newblock {\em Communications of the ACM}, 63(4):118, 2020.

\bibitem[Str80]{Stromquist80}
Walter Stromquist.
\newblock How to cut a cake fairly.
\newblock {\em The American Mathematical Monthly}, 87(8):640--644, 1980.

\bibitem[Var74]{Varian74}
Hal~R. Varian.
\newblock Equity, envy, and efficiency.
\newblock {\em Journal of Economic Theory}, 9(1):63--91, 1974.

\bibitem[Woo80]{Woodall80}
D.~R. Woodall.
\newblock Dividing a cake fairly.
\newblock {\em Journal of Mathematical Analysis and Applications}, 78(1):233--247, 1980.

\end{thebibliography}

\appendix

\section{A non-degenerate counterexample}\label{sec:ndg}

Another special case of instances with nonzero valuations are those in which every agent has distinct values for distinct subsets of goods. Such instances are called \emph{non-degenerate}. Slight perturbations can make any instance non-degenerate without introducing any new EFX allocations, so an algorithm looking for an EFX allocation can assume its input is non-degenerate without loss of generality \cite[Lemma~1]{ChaudhuryGM20,ChaudhuryGM24}. Since the algorithm for three agents~\cite{ChaudhuryGM20} works similarly, Chaudhury \cite{RayChaudhury21Thesis} proposed an approach to finding EFX allocations (for any number of agents) which maintains a partial Pareto optimal EFX allocation and modifies it into a ``better'' partial Pareto optimal EFX allocation according to some potential. As a first step, he suggests understanding the existence of Pareto optimal EFX allocations in non-degenerate instances.

The instances in \Cref{thm:main,thm:approx} are clearly degenerate, but slight variations of them are non-degenerate. We attribute this extension to Amrit Daswaney, who shared one such example with us. The example below was further simplified by GPT-6 Astra.

Consider the following ``base'' valuations, denoted $w_i$ for agent $i$.

\begin{center}\normalfont
\setlength{\tabcolsep}{11pt}\renewcommand{\arraystretch}{1.13}
\begin{tabular}{c|rr|rrrrr|rr}
     & $x$ & $y$ & $a$ & $b$ & $c$ & $d$ & $e$ & $p$ & $q$ \\
    \hline
    1 & 6 & 5 & 2 & 2 & 0 & 0 & 2 & 1 & 0 \\
    2 & 5 & 6 & 0 & 2 & 2 & 0 & 2 & 0 & 1 \\
    3 & 6 & 5 & 0 & 0 & 2 & 2 & 2 & 1 & 0 \\
    4 & 5 & 6 & 2 & 0 & 0 & 2 & 2 & 0 & 1
\end{tabular}
\end{center}

This instance is degenerate (and even has zero values). In the actual instance, we order the goods as 
$$(g_0,g_1,\ldots,g_8)=(e,x,y,a,b,c,d,p,q)$$
and take the additive valuations defined by
\[
  \quad v_i(g_j)=2^9\,w_i(g_j)+2^j,\qquad j=0,\ldots,8.\quad
\]

Every agent has distinct values for distinct bundles by considering the values modulo $2^9$. It is not hard to check that there is no Pareto optimal EFX allocation in this instance, either programmatically or by following nearly the same line of reasoning as in the proof of \Cref{thm:main}. We leave the details to the interested reader.

\end{document}